\documentclass[letterpaper, 10pt, conference]{ieeeconf} 

\IEEEoverridecommandlockouts 

\usepackage{cite}
\usepackage{amsmath,amssymb,amsfonts}
\usepackage{algorithmic}
\usepackage{graphicx}
\usepackage{textcomp}
\usepackage{bm}
\usepackage{subcaption}
\let\proof\relax
\let\endproof\relax

\usepackage{amsthm}
\newtheorem{theorem}{Theorem}

\newtheorem{definition}{Definition}[section]
\newtheorem{remark}{Remark}
\newtheorem{lemma}{Lemma}
\newtheorem{assumption}{Assumption}

\DeclareMathOperator*{\diag}{diag}

\usepackage{lipsum}
\usepackage{titlesec}

\makeatletter
\let\NAT@parse\undefined
\makeatother

\usepackage[pdftex, plainpages=false,hypertexnames=true,pdfnewwindow=true,colorlinks=true,citecolor=blue,linkcolor=red,urlcolor=blue,filecolor=blue]{hyperref}%

\title{\LARGE \bf Liquid-Graph Time-Constant Network \\*for Multi-Agent Systems Control}

\author{Antonio Marino$^{1}$, Claudio Pacchierotti$^{2}$, Paolo Robuffo Giordano$^{2}$ 
\thanks{$^{1}$ A. Marino is with Univ Rennes, CNRS, Inria, IRISA -- Rennes, France. E-mail: antonio.marino@irisa.fr.}
\thanks{$^{2}$ C. Pacchierotti and P. Robuffo Giordano are with CNRS, Univ Rennes, Inria, IRISA -- Rennes, France. E-mail: \{claudio.pacchierotti,prg\}@irisa.fr.}
\thanks{This work was supported by the ANR-20-CHIA-0017 project ``MULTISHARED''.}
}

\begin{document}

\maketitle
\thispagestyle{empty}
\pagestyle{empty}

\begin{abstract}
In this paper, we propose the Liquid-Graph Time-constant (LGTC) network, a continuous graph neural network (GNN) model for control of multi-agent systems based on the recent Liquid Time Constant (LTC) network. We analyse its stability leveraging contraction analysis and propose a closed-form model that preserves the model contraction rate and does not require solving an ODE at each iteration. Compared to discrete models like Graph Gated Neural Networks (GGNNs), the higher expressivity of the proposed model guarantees remarkable performance while reducing the large amount of communicated variables normally required by GNNs. We evaluate our model on a distributed multi-agent control case study (flocking) taking into account variable communication range and scalability under non-instantaneous communication.
\end{abstract}


\section{INTRODUCTION}
%

Communication is a crucial element in achieving distributed solutions for multi-agent systems (MAS) from control to planning~\cite{I-Magnus2017ControlMRSsurvey}. Like many distributed control algorithms, also learning-based control can benefit from inter-agent communication to partition the prediction process across multiple machines contributing to task scalability and prediction accuracy~\cite{majcherczyk2021flow}.

%

Leveraging communication, recent trends in data-driven distributed control involve employing Graph Neural Networks (GNNs) to encode distributed control and planning solutions. In this respect, Gama et al.~\cite{gama2021graph} extend GNNs to flocking control for large teams. Additional examples of the use of GNNs for distributed control are found in space coverage~\cite{li2021message}, multi-robot path planning~\cite{tolstaya2021multi}, and motion planning~\cite{khan2020graph}, including obstacle-rich environments~\cite{ji2021decentralized}. GNNs also enhance multi-agent perception~\cite{zhou2022multi} and enable distributed active information acquisition~\cite{P-Zhou2022GNNPerception}, translating multi-robot information gathering into graph representation and formulating GNN-based decision-making.

In the recent literature, one of the objectives is to make the learning-based method robust and stable~\cite{hewing2020learning}. In this context, many works applied contraction analysis to demonstrate recurrent neural network stability~\cite{davydov2022non}, or directly closed-loop stability in continuous learning~\cite{song2022stability} and adaptive control~\cite{tsukamoto2021learning}.
Recently, the stability analyses presented in~\cite{bonassi2021stability,terzi2021learning} showcased the concepts of ISS (Input-to-State Stability) and incremental ISS ($\delta$ISS)~\cite{bayer2013discrete,d2022incremental} in the context of LSTMs and GRUs, which are two of the most popular recurrent neural network models. Inspired by these last results, we proposed~\cite{10458338} the conditions for $\delta$ISS of the recurrent version of GNN, i.e. Gated Graph Neural Networks (GGNN)~\cite{ruiz2020gated}.

In this study, we introduce a continuous graph ordinary differential equation (ODE) network called Liquid-Graph Time-Constant (LGTC) network, inspired by its single-agent counterpart, the Liquid Time Constant (LTC) network~\cite{hasani2021liquid}. The LTC network demonstrates a strong representation learning capability, empowering agents to extrapolate and make inferences even in unfamiliar scenarios. Additionally, its closed-form version~\cite{hasani2022closed} eliminates the need to solve the initial value problem of the ODE, a requirement in other ODE networks~\cite{haber2017stable, chen2018neural, lechner2019designing}. In the existing literature, there are few attempts to make graph ODEs for modelling dynamic network systems~\cite{poli2019graph}, predicting traffic flow~\cite{qin2024learning}, or for sequential recommendations~\cite{su2022graph}. These models typically employ an auto-regressive network over a vector field modelled by a Graph Neural Network (GNN), initializing the internal state with the previous layer's output. In contrast, the proposed LGTC aims to introduce a novel model that learns a dynamic (or ``liquid") time constant, influenced by both input data and the internal state communicated over the agent graph. This dynamic time constant enables each element of the state to capture specific dynamics, thereby enhancing the model's predictive capabilities. The specialized dynamics can be used to reduce the number of the network's internal state (memory) and consequently the number of information each agent needs to communicate. The final aim of this study is to approach model-based algorithms' communication efficiency while increasing prediction accuracy, which in control translates into a higher distributed control efficiency.

In addition, our contribution includes proposing an approximated closed-form for the LGTC system, enhancing the computational efficiency and reducing the communication load. We proceed to evaluate the proposed models alongside GGNN and GraphODE in a flocking control example.
\section{PRELIMINARIES}
\label{Preliminaries}
Let $\mathcal{G}=(\mathcal{V},\,\mathcal{E})$ be an undirected graph where $\mathcal{V} = \{v_1, \dots, v_N\}$ is the vertex set (representing the $N$ agents in the group) and $\mathcal{E}\subseteq \mathcal{V}\times \mathcal{V}$ is the edge set. 
Each edge $e_k=(i,\,j)\in\mathcal{E}$ is associated with a weight $w_{ij}\geq 0$ such that $w_{ij}>0$ if the agent $i$ and $j$ can interact and $w_{ij}=0$ otherwise. 
As usual, we denote with $\mathcal{N}_i=\{j\in\mathcal{V}|\;w_{ij}>0\}$ the set of neighbors of agent $i$.
We also let $\bm{A} \in \mathbb{R}^{N\times N}$ be the adjacency matrix with entries given by the weights $w_{ij}$. Defining the degree matrix $\bm{D} =\diag(d_i)$ with $d_i = \sum_{j\in \mathcal{N}_i} w_{ij}$, the Laplacian matrix of the graph is $\bm{L}=\bm{D}-\bm{A}$.

The graph signal $\bm{x} \in \mathbb{R}^N$, whose $i$-component $x_i$ is assigned to agent $i$, can be processed over the network by the following linear combination rule applied by each agent
\begin{equation}
    \small
    \boldsymbol s_i\bm{x} = \sum_{j\in \mathcal{N}_i} s_{ji} (x_i - x_j),
    \label{eq:aggregation}
\end{equation}
where $\boldsymbol s_i$ is the $i$-th row of $S$. The signal manipulation can be operated by means of any \textit{graph shift operator} $S \subseteq \mathcal{S}$, e.g., Laplacian, adjacency matrix, weighted Laplacian , which respects the sparsity pattern of the graph. Later, in Sect.~\ref{flocking-control-example}, we will use the Laplacian as support matrix, as it is commonly used in distributed control. However, the proposed techniques do not assume the use of a specific support matrix.

Performing $k$ repeated applications of $S$ on the same signal represents the aggregation of the $k$-hop neighbourhood information. In analogy with traditional signal processing, this property can be used to define a linear graph filtering~\cite{P-Shuman2013SPG} that processes the multi features signal $\bm{x} \in \mathbb{R}^{N \times G}$ with $G$ features:
\begin{equation}
\small
    H_{S}(\bm{x}) = \sum_{k=0}^K S^k \bm{x} \bm{H_k}
    \label{eq:filter}
\end{equation}
where the weights $\bm{H_k} \in \mathbb{R}^{G \times F}$ define the output of the filter. Note that $S^k=S(S^{k-1})$, so that it can be computed locally with repeated 1-hop communications between a node and its neighbors. Hence, the computation of $H_{\bm{S}}$ is distributed on each node.
\subsection{Graph Neural Network}
\label{P-GNN}
Although $H_{\bm{S}}$ is simple to evaluate, it can only represent a linear mapping between input and output filters. GNNs increase the expressiveness of the linear graph filters by means of pointwise nonlinearities $\rho : \mathbb{R}^{N \times F_{l-1}} \rightarrow{} \mathbb{R}^{N \times F_{l-1}}$ following a filter bank. Letting $H_{\bm{S}l}$ be a bank of $F_{l-1} \times F_l$ filters at layer $l$, the GNN layer is defined as
\begin{equation}
    \bm{x}_l = \rho(H_{\bm{S}l}(\bm{x}_{l-1})), \qquad  \bm{x}_{l-1} \in \mathbb{R}^{N \times F_{l-1}}.
\end{equation}
Starting by $l=0$ with $F_0$, the signal tensor $\bm{x}_{l_n} \in \mathbb{R}^{N \times F_{l_n}}$ is the output of a cascade of $l_n$ GNN layers. This specific type of GNN is commonly referred to as a convolution graph network because each layer utilizes a graph signal convolution~\eqref{eq:filter}. By the use of imitation learning or similar techniques, the GNN can learn a distributed policy by finding the optimal filter weights $\bm{H_k}$ to propagate information among the agents and generate the desired output. Notably, each agent employs an identical version of the network and exchanges intermediate quantities with the other agents in the network through the application of $S$, resulting in an overall distributed neural network.
GNNs inherit some interesting properties from graph filters, such as permutational equivariance~\cite{gama2020stability} and their local and distributed nature, showing superior ability to process graph signals~\cite{P-Zhou2022GNNPerception,I-GamaRibeiro2020GNNPathPlanning,tolstaya2021multi}.  

\subsection{Gated Graph Neural Network}
\label{P-GGNN}
Recurrent models of GNNs can solve time-dependent problems. These models, similarly to recurrent neural networks (RNNs), are known as graph recurrent neural networks (GRNNs). GRNNs utilize memory to learn patterns in data sequences, where the data is spatially encoded within graphs, regardless of the team size of the agents~\cite{gu2020implicit}. However, traditional GRNNs encounter challenges such as vanishing gradients, which are also found in RNNs. Additionally, they face difficulties in handling long sequences in space, where certain nodes or paths within the graph might be assigned more importance than others in long-range exchanges, causing imbalances in the graph's informational encoding. \\
Forgetting factors can be applied to mitigate this problem, reducing the influence of past or new signal on the state. A Gated Graph Neural Network (GGNN)~\cite{ruiz2020gated} is a recurrent Graph Neural Network that uses a gating mechanism to control how the past information influences the update of the GNN states. We can add a state and an input gates, $\bm{\hat{q}}, \bm{\tilde{q}} \in Q \subseteq [0,1]^{N \times F}$, that are multiplied via the Hadamard product~$\circ$ by the state and the inputs of the network, respectively. These two gates regulate how much the past information and the input are used to update the network's internal state. GGNNs admit the following state-space representation~\cite{lukovnikov2020improving},
\begin{equation}
\small
    \begin{cases}
    \tilde{\bm{q}} = \sigma(\tilde{A}_S(\bm{x}) + \tilde{B}_S(\bm{u}) +\hat{b}) \\
    \hat{\bm{q}} = \sigma(\hat{A}_S(\bm{x}) + \hat{B}_S(\bm{u}) +\tilde{b})  \\
    \bm{x}^+ = \sigma_c(\bm{\hat{q}}\circ A_S(\bm{x}) + \bm{\tilde{q}} \circ B_S(\bm{u}) +b)
    \end{cases}
    \label{system}
\end{equation}
with $\sigma(x) = \frac{1}{1+e^{-x}}$ being the logistic function, and $\sigma_c(x) = \frac{e^{x}-e^{-x}}{e^{x}+e^{-x}}$ being the hyperbolic tangent. $\hat{A}_S,\hat{B}_S$ are graph filters of the forgetting gate, $\tilde{A}_S,\tilde{B}_S$ are graph filters~\eqref{eq:filter} of the input gate, and $A_S$ and $B_S$ are the state graph filters~\eqref{eq:filter}. $\bm{\hat{b}},\bm{\tilde{b}},\bm{b} \in \mathbb{R}^{N \times F}$ are respectively the biases of the gates and the state built as $\bm{1}_N \otimes \bm{\textit{b}}$ with the same bias for every agents. 
We consider the system under the following assumption
\begin{assumption}
	The input $\bm{u}$ is unity-bounded: $\bm{u} \in \mathcal{U} \subseteq [-1,1]^{N \times G}$ , i.e. $||\bm{u}||_{\infty} \leq 1$.
	\label{assumption1}
\end{assumption}
\noindent We identify the induced $\infty$-norm as $||\cdot||_{\infty}$. We used the following notation for the filters in the system~\eqref{system}:
\begin{equation}
	\small
	\begin{aligned}
		& \qquad S_{I,K} \triangleq [I, S,\dots, S^K ] \\ 
		A_{0,K}   \triangleq [A_0, \dots, A_K ]^T & \qquad
		B_{0,K}  \triangleq [B_0, \dots, B_K ]^T \\
		\tilde{A}_{0,K} \triangleq [\tilde{A}_0, \dots, \tilde{A}_K]^T & \qquad
		\hat{A}_{0,K}  \triangleq [\hat{A}_0, \dots, \hat{A}_K]^T  \\
		\tilde{B}_{0,K}  \triangleq [\tilde{B}_0, \dots, \tilde{B}_K ]^T & \qquad
		\hat{B}_{0,K} \triangleq [\hat{B}_0, \dots, \hat{B}_K ]^T
	\end{aligned}
	\label{eq:definitions}
\end{equation}
where $K$ is the filters length. Then, in light of assumption~\ref{assumption1} and knowing that $ ||\bm{x}||_{\infty}\leq 1$,  each gate feature $q_i$ satisfies:
\begin{equation}
	\small
	\begin{aligned}
		|\hat{q}_i| & \leq 
		\sigma( ||S_{I,K}||_{\infty}(|| \hat{A}_{0,K} ||_{\infty}+ ||\hat{B}_{0,K}||_{\infty})+||\hat{\bm{b}}||_\infty) \triangleq \sigma_{\hat{q}}.
	\end{aligned}
	\label{eq:gates}
\end{equation}

\begin{assumption}
	Given any two support matrices $||S_{1}(t)||_{\infty}$ and $||S_{2}(t)||_{\infty}$, $\forall t \in \mathbb{R}^{+}$ associated with two different graphs, they are bounded by the same $||\bar{S}||_{\infty}$; moreover, they are lower bounded by $||\tilde{S}||_{\infty}$.
    \label{assumption2}
\end{assumption}
\noindent In our previous work~\cite{10458338}, we derived the following stability condition
\begin{theorem}
	\label{dISS_stab}
	Under Assumptions~\ref{assumption1} and~\ref{assumption2}, a sufficient condition for the system~\eqref{system} to be $\delta$ISS is $\mathcal{A}_{\delta} \leq 1$; where
	\small
	\begin{flalign}
		\begin{aligned}
			\mathcal{A}_{\delta} & \triangleq \sigma_{\hat{q}}||\bar{S}_{I,K}||_\infty ||A_{0,K}||_{\infty}+\frac{1}{4}||\bar{S}_{I,K}||^2_\infty||\hat{A}_{0,K}||_{\infty} ||A_{0,K}||_{\infty} \\ & + \frac{1}{4} ||\bar{S}_{I,K}||^2_\infty||\tilde{A}_{0,K}||_{\infty} ||B_{0,K}||_{\infty}.
		\end{aligned}
	\end{flalign}
	\normalsize
    Where $\bar{S}_{I,K}$ is the defined as in eq. \eqref{eq:definitions}
\end{theorem}

\begin{remark}
	\label{remark-nL}
	In practice, Assumption~\ref{assumption2} is met by limiting the cardinality of $\# \mathcal{N}_i < N$ for every agent $i$ in the team. For normalized Laplacian, the assumption is met without any further restrictions on the graph topology. 
\end{remark}

\section{LIQUID-GRAPH TIME CONSTANT NETWORK}
\label{sec:lgtc}
In light of the recent findings on the continuous time neural ODE, we want to propose a novel continuous-time graph neural network: Liquid-Graph Time-Constant (LGTC) network. Inspired by the LTC network proposed in~\cite{hasani2021liquid}, the LGTC is described by the following ODE:
\begin{equation}
\small
    \begin{cases}
	\bm{f} = \rho(\hat{A}_S(\bm{x})+\bm{b_x}) + \rho(\hat{B}_S(\bm{u})+\bm{b_u}) \\
	\dot{\bm{x}} = -(\bm{b} + \bm{f}) \circ \bm{x} - \sum_{k=1}^K \bm{S}^k \bm{x} \bm{A_k}. +\bm{f} \circ \sigma_c(B_{S}(\bm{u}))	
    \end{cases}
\label{eq:sys-lgtc}
\end{equation}
where {\small $\rho(x)=ReLU(x)$}. Like in GGNN system, \small $\hat{A}_S,\hat{B}_S, A_S, B_S$ \normalsize are graph filters~\eqref{eq:filter} and \small $\bm{b_x},\bm{b_u}, \bm{b} \in \mathbb{R}^{N \times F}$ \normalsize are biases of the model. This model disposes of an input-state-dependent varying time-constant allowing single elements of the hidden state to identify specialized dynamical systems for input features arriving at each time step.

In the following results on the system, we used the vector operator $\bm{X}_{|}$ that rearranges the elements of matrix $\bm{X}$ in a vector. For the system in~\eqref{eq:sys-lgtc}, the following lemma holds
\begin{lemma}
	\label{state_bound}
		If $b_i>0$ and {\small $\bm{A_k}^T \otimes \bm{S}^k \geq 0$}, for $k=[1,\dots,K]$, the state feature $i$ at time $t$ is bounded in the range $[-1,1]$ if \small $\bm{x(0)} \in \mathcal{X} \subseteq[-1,1]^{N \times F}$ \normalsize. 
\end{lemma}
\proof
	Defining the following Lypunov function 
	\begin{equation*}
		\small
		V = \frac{1}{2}\bm{x_{|}}^T\bm{x_{|}}
	\end{equation*}
	Its time derivative along the system trajectories is
	\small 
	\begin{equation*}
		\begin{split}
		\dot{V} &= -\bm{x}_{|}^T \diag(\bm{b}_{|}) \bm{x}_{|} -\bm{x}_{|}^T \diag(\bm{f}_{|}) \bm{x}_{|} \\ & -\bm{x}_{|}^T \sum_{k=1}^K \bm{A_k}^T \otimes \bm{S}^k \bm{x}_{|} +\bm{x}_{|}^T \diag(\bm{f}_{|})\sigma_c(B_{S}(\bm{u}))_{|}
		\end{split}
	\end{equation*}
	\normalsize
	By imposing $b_i\geq0$ and \small $\sum_{k=1}^K \bm{A_k}^T \otimes \bm{S}^k \geq 0$ \normalsize for \small $k=[1,\dots,K]$, \normalsize we have
	\begin{equation*}
        \small
		\dot{V} < -\bm{x}_{|}^T \diag(\bm{f}_{|}) \bm{x}_{|} +\bm{x}_{|}^T \diag(\bm{f}_{|})\sigma_c(B_{S}(\bm{u}))_{|}
    \end{equation*}
	 When {\small$x_i>1$}, {\small $\dot{V_i}<0$}, because the term {\small $\bm{x}_{|}^T\diag(\bm{f}_{|}) \sigma_c(B_{S}(\bm{u}))_{|} \leq \bm{1}^T\diag(\bm{f}_{|})\bm{1}$} and reminding that $f_i>0$. The proof is completed by noting that {\small $\dot{V_i}<0$} when {\small $x_i<-1$} with a similar reasoning.  
\endproof

\noindent We denote the vector field of the dynamic system~\eqref{eq:sys-lgtc} with \small $F(x,u,S,t): \mathcal{X} \times \mathcal{U} \times \mathcal{S} \times \mathbb{R}_{\geq 0} \rightarrow \mathbb{R}^{N\times F}$ \normalsize and its jacobian with \small $D_xF = \frac{\partial F_{|}}{\partial \bm{x}_{|}}$. \normalsize Giving the induced infinite log-norm  definition \small$\mu_{\infty}(\bm{X}) = \max_{i}(x_{ii}+\sum_{j=1,j\neq 1}^{n}|x_{ij}|)$~\cite{strom1975logarithmic}\normalsize, the followings results are known~\cite{davydov2022non}:

\begin{definition}
	\label{contraction}
	The vector field $F$ is \textit{strongly infinitesimally contracting} if:
	\begin{equation}
		\small
		\mu_{\infty}(D_xF) < -c
		\label{eq:contraction}
	\end{equation}
	 With $c>0$ known as contraction rate.
\end{definition}

\begin{definition}
	\label{diss-contractive}
	For a vector map $F$ satisfying~\ref{contraction} and given $l_u,\,l_S$ the Lipschitz constants of $F$ in the input $\bm{u}$ and $\bm{S}$,  any two solutions $\bm{x_1}$ and $\bm{x_2}$ with initial conditions $\bm{x_1(0)}$, $\bm{x_2(0)}$ and inputs $\bm{u_1(t)}$, $\bm{u_2(t)}$ satisfy the $\delta$ISS relationship
	\begin{equation}
		\small
		\begin{split}
		|| \bm{x_1(t)} - \bm{x_2(t)} ||_{\infty} &\leq e^{-ct} ||\bm{x_1(0)} - \bm{x_2(0)} ||_{\infty} \\+& \frac{l_u}{c}(1-e^{-ct}) sup_{\tau \in [0,t]}|| \bm{u}_1(\tau) -\bm{u}_2(\tau) ||_{\infty} \\ + & \frac{l_S}{c}(1-e^{-ct}) sup_{\tau \in [0,t]}|| S_1(\tau) - S_2(\tau) ||_{\infty}
	 	\end{split}
	\end{equation} 
\end{definition}

\noindent For the neural network~\eqref{eq:sys-lgtc}, we provide the following theorem
\begin{theorem}
	\label{diss-lgtc}
	Under Assumption~\ref{assumption2}, with $\bm{x(0)} \in \mathcal{X}$, system~\eqref{eq:sys-lgtc} is $\delta$ISS if the following constraints are satisfied
	\begin{equation}
		\small
			c \geq 0; \quad  \bm{b}\geq 0; \quad  \mu_{\infty}(\sum_{k=1}^K \bm{A_k}^T \otimes \bm{S}^k) \geq 0 
	\end{equation}
	\begin{equation}
		\small
		\begin{split}
			\text{with} \qquad c &= ||\bm{b}||_{\infty} +||A_{1,K}||_{\infty}||\tilde{S}_{1,K}||_{\infty} \\ & +||\bm{b_x}||_{\infty} - ||\hat{A}_{0,K}^T||_{\infty}||\bar{S}_{0,K}||_{\infty}	
		\end{split}.
	\end{equation}
\end{theorem}
\proof
Let be {\small $\sigma_{u}= \sigma_c(B_S(\bm{u}))\subseteq[-1,1]^{N \times F}$}, the log-norm of the jacobian $\mu_{\infty}(D_xF)$ has its superior in:
\begin{equation*}
	\small
	\begin{split}
		\sup\limits_{x} \mu_{\infty}(D_xF) = & \sup\limits_{x}  \mu_{\infty}(-\diag(\bm{b_{|}})-\diag(\bm{f}_{|})\\ -\sum_{k=1}^{K} \bm{A_k}^T \otimes \bm{S}^k + \diag(&\sigma_{u|}- \bm{x}_{|})\diag(D_x\bm{f}) \sum_{k=0}^{K} \bm{\hat{A}_k}^T \otimes \bm{S}^k ) \\
		\leq -||\bm{b}||_{\infty}+\mu_{\infty} (&- \sum_{k=1}^{K} \bm{A_k}^T \otimes \bm{S}^k ) + \sup\limits_{x} [-||\bm{f}_{|}||_{\infty} \\ + \mu_{\infty}( \diag(&\sigma_{u|}-\bm{x}_{|} )\diag(D_x\bm{f})\sum_{k=0}^{K} \bm{\hat{A}_k}^T \otimes \bm{S}^k)] 
	\end{split}
\end{equation*}
where \small$D_x\bm{f} \in [0,1]^{NF}$ \normalsize is the derivative of \small$\bm{f}$ \normalsize in \small $\bm{x}_{|}$\normalsize. The last inequality follows by \small $\mu_{\infty}(\bm{X}) = ||\bm{X}||_{\infty}$ \normalsize when \small$\bm{X}>0$ \normalsize and by the fact that {\small $\bm{b}_{|} \geq 0$} and {\small $\bm{f}_{|}\geq0$}. Moreover, we can note that \small $ -|| \bm{f}_{|} ||_{\infty} \leq - ||\rho(\hat{A}_S(\bm{x})+\bm{b_x})_{|}||_{\infty}$\normalsize. Giving {\small$\sigma_{u|} \in [-1,1]^{NF}$} and {\small$-\mu_{\infty}(\bm{A_k}^T \otimes \bm{S}^k ) \leq 0$}, by lemma~\ref{state_bound}, {\small $\bm{x}_{|}\in [-1,1]^{NF}$}. Therefore, we get 
\begin{equation*}
	\small
	\begin{split}
		\mu_{\infty}(D_xF) \leq & -||\bm{b}||_{\infty}-||\sum_{k=1}^{K} \bm{A_k}^T \otimes \bm{S}^k||_{\infty} - \\  ||\rho(\hat{A}_S(\bm{x}+\bm{b_x}))&||_{\infty}+\max\limits_{d \in [-1,1]^{NF}} 2 \mu_{\infty}(\diag(d) \sum_{k=0}^{K} \bm{\hat{A}_k}^T \otimes \bm{S}^k) \\
		\leq -||\bm{b}||_{\infty} - ||&\sum_{k=1}^{K} \bm{A_k}^T \otimes \bm{S}^k||_{\infty} - ||\rho(\hat{A}_S(\bm{x}+\bm{b_x}))||_{\infty} + \\ 2 \max ( \mu_{\infty}(&-\sum_{k=0}^{K} \bm{\hat{A}_k}^T \otimes \bm{S}^k),\mu_{\infty}(\sum_{k=0}^{K} \bm{\hat{A}_k}^T \otimes \bm{S}^k))
	\end{split}
\end{equation*}
To further reduce the previous inequality, we can use the sub multiplicity property of the infinite norm and the fact that the last term in the previous inequality is {\small $ \leq ||\sum_{k=1}^{K} \bm{\hat{A}_k}^T \otimes \bm{S}^k||_{\infty}$} to derive
\begin{equation*}
	\small
	\begin{split}
		 \mu_{\infty}(D_xF) \leq & -||\bm{b}||_{\infty} - ||A_{1,K}||_{\infty}||\tilde{S}_{1,K}||_{\infty} \\ & -||\bm{b_x}||_{\infty} +   ||\hat{A}_{0,K}^T||_{\infty}||\bar{S}_{0,K}||_{\infty} = - c
	\end{split}
\end{equation*}
From the statement in the Theorem~\ref{diss-lgtc}, the system is contractive under the defintion \ref{contraction} and $\delta$ISS under the definition \ref{diss-contractive} with
\begin{equation}
	\small
	\begin{split}
		l_u = & (2||\hat{B}_{0,K}||_{\infty} +(||\hat{A}_{0,K}||_{\infty}||\bar{S}_{0,K}||_{\infty} + ||\hat{B}_{0,K}||_{\infty}||\bar{S}_{0,K}||_{\infty} \\ &+ ||\bm{b_x}||_{\infty} +||\bm{b_u}||_{\infty})||B_{0,K}||_{\infty} )||\bar{S}_{0,K}||_{\infty} \\
		l_S =&\binom{k+1}{2}(((||\hat{A}_{0,K}||_{\infty}||\bar{S}_{0,K}||_{\infty} +  ||\hat{B}_{0,K}||_{\infty}||\bar{S}_{0,K}||_{\infty} \\ & + ||\bm{b_x}||_{\infty} + ||\bm{b_u}||_{\infty})||B_{1,K-1}||_{\infty}+2||\hat{A}_{1,K-1}||_{\infty}\\ & + 2||\hat{B}_{1,K-1}||_{\infty})||\bar{S}_{1,K-1}||_{\infty} - ||\tilde{S}_{1,K-1}||_{\infty}||A_{1,K-1}||_{\infty}) 
	\end{split}
\end{equation}
\noindent We omitted the derivation of {\small$l_u$} and {\small$l_S$} which can be found by computing {\small $\sup_{\bm{u}}||D_uF||_{\infty}, \sup_{\bm{S}}||D_sF||_{\infty}$}, respectively. 
\endproof
The higher expressivity of the LTC model has been proved in~\cite{hasani2021liquid} via principal component analysis on trajectory depth. The results of this analysis showed that LTC performs at worst like a discrete RNN model over short trajectories, but its expressive power does not decrease on longer trajectories. We can derive the same conclusion for LGTC, which shares with LTC the main components but over a network of agents. This higher expressivity allows us to reduce the communicated state variables. Therefore, we can select a subset of state features {\small $ F' << F$} and a subset of input features {\small $ G' << G$} to communicate. Denoting by {\small $\bm{x_{F'}}$} the features to communicate and by {\small $\bm{x_{F-F'}}$} the features to not communicate, we can encode the reduced communication in the equations above by changing the filter bank in eq~\eqref{eq:filter} as {\small $\sum_{k=0}^K [\bm{I_N}\bm{x_{F-F'}}, \bm{S}^k\bm{x_{F'}}] \bm{H_k}$} with {\small $ \bm{I}_N$} the identity matrix of dimension $N$. With these reduced graph filters, Theorem~\ref{diss-lgtc} still holds true.

To compute the state of the system in eq.~\eqref{eq:sys-lgtc} at any time point $T$, the neural network must implement and ODE solver, like Runga-Kutta (RK), that simulates the system starting from a trajectory $\bm{x}(0)$, to $\bm{x}(T)$. When used in a prediction framework, the neural layer~\eqref{eq:sys-lgtc} will forward the state $\bm{x}(T)$ to the next layers. Since the system is stiff, the authors in~\cite{hasani2021liquid} proposed an hybrid solver for LTC networks that has the additional advantage of reducing the vanishing gradient problem due to the ODE resolution. We modified this solver to handle our network as follows:
\begin{equation}
	\small
	\begin{aligned}
		\bm{x(t+\Delta)} &= \frac{\bm{x(t)} +\Delta(- \sum_{k=1}^{K} \bm{S^k}\bm{x(t)}\bm{A_k} + f \circ \sigma_c(B_S(\bm{u}))}{1 + \Delta(b + f)}
	\end{aligned}.
	\label{eq:hybrid-solver}
\end{equation}
With $\Delta=T/n$ being a fixed step size for $n$ evaluations of~\eqref{eq:hybrid-solver}, starting from $t=0$ until reaching $t=T$. The support matrix $\bm{S}$ and the input $\bm{u}$ are fixed in the simulation interval $[0, T]$. Each iteration of ~\eqref{eq:hybrid-solver} requires to communicate again the network state, e.g. $\bm{x(t+\Delta)}$ to evaluate $\bm{x(t+2\Delta)}$. Therefore, solving the ODE nullifies the reduction in the communicated variables for the high expressivity of the model. For this reason, we provide in the next section a closed-form solution for LGTC that does not require an ODE solver, thus significantly gaining in computation time and communication load.  

\section{CLOSED-FORM APPROXIMATION}
\label{sec:cf-approximation}
 By fixing the input and support matrix between the interval $t=[0,T]$, the ODE solution should satisfy the following integral:
\begin{equation}
	\small
		\bm{x(T)} = \bm{x(0)} + \int_{0}^{T}F(\bm{x(\tau)},\bm{S(0)},\bm{u(0)})d \tau
	\label{eq:f-int} 
\end{equation}
However, solving this integral is not straightforward due to the non-linearity in $F$. We aim to find an approximate form of eq.~\eqref{eq:f-int} that keeps the same dynamic properties of the original system. Given the following quantities 
\begin{equation*}
	\small
	\begin{split}
		\bm{f}_i =& -D_x\bm{f}\diag(\bm{x(T)}_{|})\sum_{k=0}^{K} \bm{\hat{A}_k}^T \otimes \bm{S(0)}^k +\sum_{k=1}^{K} \bm{A_k}^T \otimes \bm{S(0)}^k \\
		\bm{f}_x =& \rho(\hat{A}_S(\bm{x(T)})+\bm{b_x}) \\
		\bm{f}_{\sigma} =& \rho(\hat{B}_S(\bm{u(0)}) + \bm{b_u}) + \rho(\hat{A}_S(\bm{x(0)}) + \bm{b_x}) \\
	\end{split}
\end{equation*}
we propose the following closed-form approximation
\begin{equation}
	\label{eq:cgfc}
	\small
	\begin{split}
		\bm{x(T)_{|}} =& e^{-(\diag(\bm{b}_{|}) +\diag(\bm{f}_{x|}) + \bm{f}_i)T} \diag(\sigma(2\bm{f}_{\sigma})_{|})\bm{x(0)_{|}} \\ & + \diag(1-\sigma(2\bm{f}_{\sigma})_{|})\sigma_{u|}.
	\end{split}
\end{equation} 
Reminding that \small $\sigma_{u}= \sigma_c(B_S(\bm{u(0)}))$, \normalsize one can easily demonstrate that \small $\bm{x(T)} \in \bm{\mathcal{X}}$ \normalsize since its expression is a convex combination with coefficient \small $\sigma(2\bm{f}_{\sigma})$ \normalsize and \small $1-\sigma(2\bm{f}_{\sigma})$. \normalsize This closed-form solution maintains the contraction rate as stated by the following theorem.
\begin{theorem}
	\label{closed-form}
	With $t=[0,T]$, the state trajectory~\eqref{eq:cgfc} is a closed-form approximation of the dynamic system~\eqref{eq:sys-lgtc} with the same contraction rate $c$ of the Theorem~\ref{diss-lgtc}.
\end{theorem} 
\begin{proof}
The theorem can be demonstrated by showing that the closed-loop system maintains the same dynamic properties. We start the proof by highlighting the relationship between the superior of the induced log-norm and the superior of the time derivative of the induced norm of the matrix exponential  
\begin{equation}
	\small
	\sup\limits_x \mu_{\infty}(D_xF)= \frac{\partial}{\partial t} \sup\limits_{x} || \exp(D_xFt) ||_{\infty}\Big|_{t=0^+}
\label{eq:equal-c}
\end{equation}
Therefore, we can show that the equivalence between the right-hand side of the previous relation and the induced norm of the jacobian \small ($D_x\bm{x(T)}$) \normalsize of eq.~\eqref{eq:cgfc}. Let's name $F_{ex}$ the argument of the exponential in the eq.~\eqref{eq:cgfc}. \small $||D_x\bm{x(t)}||_{\infty}$ \normalsize satisfies the following :
\begin{equation}
	\small
	\begin{split}
	||D_x\bm{x(t)_{|}} ||_{\infty} &= ||e^{F_{ex}} \diag(\sigma(2f_{\sigma})_{|}) + \\ & \diag(\bm{x}_{|})e^{F_{ex}}(\diag(\sigma(2f_{\sigma})_{|})\sum_{k=0}^{K} \bm{\hat{A}_k}^T \otimes \bm{S}^k t + \\ &\diag(\bm{x}_{|})e^{F_{ex}}2\diag(D_{\sigma |})\sum_{k=0}^{K} \bm{\hat{A}_k}^T \otimes \bm{S}^k -  \\ & \diag(\sigma_u)2\diag(D_{\sigma|})\sum_{k=0}^{K}\bm{\hat{A}_k}^T \otimes \bm{S}^k||_{\infty} \\
	&\leq ||e^{F_{ex}}||_{\infty} + ||e^{F_{ex}}\sum_{k=0}^{K} \bm{\hat{A}_k}^T \otimes \bm{S}^k t||_{\infty} + \\ & 0.5(||e^{F_{ex}}||_{\infty} + 1)||\sum_{k=0}^{K} \bm{\hat{A}_k}^T \otimes \bm{S}^k||_{\infty} \\
    \leq ||e^{F_{ex}}||_{\infty} + &||e^{F_{ex}}\sum_{k=0}^{K} \bm{\hat{A}_k}^T \otimes \bm{S}^k t||_{\infty} + ||\sum_{k=0}^{K} \bm{\hat{A}_k}^T \otimes \bm{S}^k||_{\infty}
	\end{split}
\end{equation}
where the last inequalities come from similar reasons to the proof in the theorem~\ref{diss-lgtc} and the Lipschitz constant of $\sigma(2x)$ equal to $0.5$. The following time derivative proves the equality~\eqref{eq:equal-c} 
\begin{equation}
 	\small
 	\begin{split}
 	 \frac{\partial}{\partial t}\sup\limits_x||D_x\bm{x(t)_{|}}||_{\infty}\big|_{t=0^{+}} =& -||\bm{b}||_{\infty} - ||A_{1,K}||_{\infty}||\tilde{S}_{1,K}||_{\infty} - \\ & ||\bm{b_x}||_{\infty} +   ||\hat{A}_{0,K}^T||_{\infty}||\bar{S}_{0,K}||_{\infty} 
\end{split}
\end{equation}
Moreover, one can easily demonstrate that the Lipschitz constant in \small $\bm{u}$ \normalsize and \small $\bm{S}$ \normalsize is upper bounded by $l_u,l_S$ respectively, taking into account that $te^{-(F)t}; F>0$ can be upper bounded by $1$.  
\end{proof}

With a fixed $T$, the state trajectory~\eqref{eq:cgfc} can serve as a discrete-time system that in training and evaluation requires one iteration/communication to be evaluated, against the $n$ iteration required by the continuous system. However, equation~\eqref{eq:cgfc} can not be directly distributed over the agents due to the kronacker product terms appearing in $f_i$. Moreover, notably, the exponential term in recursion causes vanishing gradient problems in training. Therefore, we modified the trajectory~\eqref{eq:cgfc} to tackle these two issues. First, we multiply the kronacker product terms by $\bm{x}/(\bm{x}+\epsilon)$, with $/$ being the division element by element and a small $\epsilon$. Second, we replace the exponential term with a sigmoidal nonlinearity which decreases much more smoothly. The resulting closed-form graph time constant (CfGC)  system is: 
\begin{equation}
	\small
	\begin{cases}
		f_{\sigma} = \rho(\hat{B}_S(\bm{u}) + \bm{b_u}) + \rho(\hat{A}_S(\bm{x})+\bm{b_x}) \\
		f_x = \rho(\hat{A}_S(\bm{x})+\bm{b_x}) \\
		f_i = -D_xf\hat{A}_S(\bm{x})+\sum_{k=1}^K \bm{S}^k \bm{x} \bm{A_k}/(\bm{x}+\epsilon) \\
		\bm{x}^+ = (\bm{x}\circ\sigma(-(\bm{b}+f_x+f_i)t + \pi)-\sigma_u)\circ\sigma(2f_{\sigma})+\sigma_u
	\end{cases}
	\label{eq:cgfc-sys}
\end{equation}
This system still satisfies the theorem~\ref{closed-form} for $\epsilon \rightarrow 0$ since $\sigma(-Ft+\pi)\big|_{t=0} \approx 1$. The model in~\eqref{eq:cgfc-sys} is different from the one in~\eqref{eq:sys-lgtc}, even if, when they have the same weight matrices, they have approximately the same dynamic properties. We demonstrate this good approximation in the result Section~\ref{flocking-control-example}.

\section{VALIDATION EXAMPLE}
\label{flocking-control-example}
We experimentally validated the neural networks proposed in a flocking control example. Despite its simplicity, this scenario has been previously studied~\cite{tolstaya2020learning} 
 using different approaches such as GNN, LSTM, allowing for a direct comparison with the proposed model. We also explicitly compare GraphODE~\cite{poli2019graph}, GGNN, LGTC and CfGC, satisfying the conditions in Theorem~\ref{dISS_stab},~\ref{diss-lgtc} and~\ref{closed-form} respectively. The stability condition is imposed by the following regularization in addition to the training loss:
\begin{equation}
    \Pi =  \sum_i \text{Softplus}(p_i)
    \label{eq:reg_stab}
\end{equation}
With $p_i$ being only one element equal to $\delta \mathcal{A}_{i} - 1$ to satisfy the Theorem~\ref{dISS_stab} and equal to $[-c,-\bm{b},-\mu_{\infty}(\sum_{k=1}^K \bm{A_k}^T \otimes \bm{S}^k)]$ for Theorem~\ref{diss-lgtc} and ~\ref{closed-form}. The Softplus is introduced to have a smooth ReLU function that can be regulated through its $\beta$ to enforce a higher contraction rate and consequently fast convergence. We used $\beta = 10$.
\begin{figure}[t]
    \centering
    \includegraphics[scale=0.35]{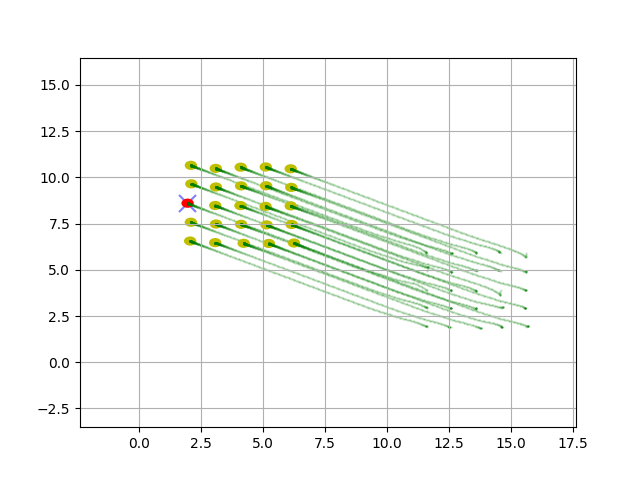}
    \caption{\textbf{Flocking control}: a group of agents (yellow dots) move in order to reach the same velocity and to avoid collision. The leader (red dot) moves in order to reach the target (blue cross) and avoids the collision with the other agents.}
    \label{fig:flocking_example}
\end{figure}

In the following, we show a case study involving flocking control (Fig.~\ref{fig:flocking_example}) with a leader. In this problem, the agents are initialized to follow random velocities while the goal is to have them all travelling at the same velocity while avoiding collisions with each other. Moreover, one of the agents takes the leader role, conducting the team toward a target unknown to the other agents. Flocking is a canonical problem in decentralized robotics~\cite{yu2010distributed, beaver2021overview}.

We considered $N$ agents described by the position $\bm{r}(t) \in \mathbb{R}^{N \times 2}$ and the velocity $\bm{v}(t) \in \mathbb{R}^{N \times 2}$ with a double integrator dynamics
\begin{equation*}
        \bm{r}(t+1) = \bm{r}(t) + T\bm{v}(t); \quad 
        \bm{v}(t+1) = \bm{v}(t) + T\bm{u}(t);
\end{equation*}
with the discrete acceleration $\bm{u}(t) \in \mathbb{R}^{N \times 2}$ taken as system input.  Note that the agent dynamics is used for building the dataset and for simulation purposes, but it is not provided to the learning algorithm. The flocking expert controller~\cite{tolstaya2020learning} for the follower agents is given by 
\begin{subequations}
\begin{equation}
    \bm{u}_f(t) =  \frac{1}{N}\sum_{i=1}^{N} \bm{v_i}(t) - \nabla_{\bm{r}}CA(\bm{r}(t),\bm{r}_j(t))|_{j=1\dots N} 
\end{equation}
\text{and for the leader it is}
\begin{equation}
    \bm{u}_l(t) = -W_p(\bm{r}_l(t) - \bm{d}(t)) - \nabla_{\bm{r}_l}CA(\bm{r}_l(t),\bm{r}_j(t))|_{j=1\dots N} 
\end{equation}
\label{eq:expert_flocking}
\end{subequations}
\noindent where $W_p$ is a gain, $\bm{r}_l \in \mathbb{R}^{2}$ is the leader position, $\nabla_{\bm{r}} CA(\bm{r}(t), \bm{r}_j(t))$/$\nabla_{\bm{r}_l}CA(\bm{r}_l(t), \bm{r}_j(t))$ are the gradient of the collision avoidance potential with respect to the position of the agents/leader $\bm{r}$/$\bm{r}_l$, evaluated at the position $\bm{r}(t)$/$\bm{r}_l(t)$ and the position of every other agent $\bm{r}_j(t)$ at time $t$. The i-element of $\nabla_{\bm{r}}CA$ for each robot $i$ with respect to robot $j$ is given by \cite{tanner2003stable}
\begin{equation}
    \nabla_{\bm{r}_i} CA(\bm{r}_{ij}) = \begin{cases}
    -\frac{\bm{r}_{ij}}{||\bm{r}_{ij}||_2^4} - \frac{\bm{r}_{ij}}{||\bm{r}_{ij}||_2^2} & if ||\bm{r}_{ij}||_2^2 \leq R_{CA} \\
    \qquad \bm{0} & otherwise
    \end{cases}
\end{equation}
with $\bm{r}_{ij} = \bm{r}_i - \bm{r}_j$ and $R_{CA} > 0$ indicating the minimum acceptable distance between agents. This potential function is a non-negative, non-smooth function that goes to infinity when the distance reduces and grows when the distance exceeds $R_{CA}$, in order to avoid the team losing  connectivity~\cite{tanner2003stable}. $\bm{u}_f(t),\bm{u}_l(t)$ are a centralized controller since computing them requires agent $i$ to have instantaneous evaluation of $\frac{1}{N}\sum_{i=1}^{N} \bm{v_i}(t)$ and $\bm{r}_j(t)$ of every other agent $j$ in the team. $R_{CA}$ and $W_p$ are tunable parameters of the controllers.  
\subsection*{\textbf{Neural Network Architecture}}
We assume that the agents form a communication graph when they are in a sphere of radius $R$ between each others and that exchanges occur at the sampling time $T=0.05$~s, so that the action clock and the communication clock coincide. \\
The input features vector $\bm{w}_i \in \mathbb{R}^{10}$ of the robot $i$ for the designed neural network is 
\begin{equation}
\begin{split}
    \bm{w}_i = \Bigg[ & \bm{v}_i, \sum_{j \in \mathcal{NS}_i} \frac{\bm{r}_{ij}}{||\bm{r}_{ij}||_2^4}, \sum_{j \in \mathcal{NS}_i}\frac{\bm{r}_{ij}}{||\bm{r}_{ij}||_2^2}, \\ & \{\bm{0}_2, \bm{r}_l - \bm{d} \}, \{ [0,1],[1,0] \} \Bigg]
\end{split}
\end{equation}
where $\mathcal{NS}_i$ is the set of the sensing agents within a sphere of radius $R_{CA}$ centred in the robot $i$. Moreover, the vector $w_i$ contains the zero vector $\bm{0}_2 \in \mathbb{R}^{1 \times 2}$ and the one-hot encoding $[0,1]$, if the agent is a follower, while $(\bm{r}_l - \bm{d})$ and $[1,0]$, if the agent is a leader. We chose the one-hot encoding instead of the binary one because it allows differentiating the neural network weights between the leader and the follower. Note that we assume all the information in the vector $\bm{w}_i$ to be locally available at the sampling/control time $T$. \\
The core of the neural network for the flocking control is a layer of GGNN, GraphODE, LGTC or CfGC with $F=50$ features in the hidden state and filter length $K=2$. Note that the choice of $K$ affects the complexity of the stability condition imposed since it will constrain more parameters. For LGTC and GraphODE, we implement the hybrid solver~\eqref{eq:hybrid-solver} and RK4, respectively. The input features are first processed by a cascade of two fully connected layers of $128$ nodes before feeding the graph neural network. A readout of two layers with $128$ nodes combines the $F$-features GNN hidden state to get the bidimensional control $\bm{u}$ saturated to the maximum admissible control. The input layers and the readout that encapsulate the graph neural network shape a more realistic setting to test the stability of the graph neural network that is usually used in combination with other kinds of neural models. To reduce the communicated variables, we used $F'=4$ and $G=4$ for a subset of the state and the input to communicate. Following the Remark \ref{remark-nL}, we used normalized Laplacian as a support matrix. 
\subsection*{\textbf{Training}}
\begin{figure*}[t]
\centering
    \includegraphics[width=\textwidth]{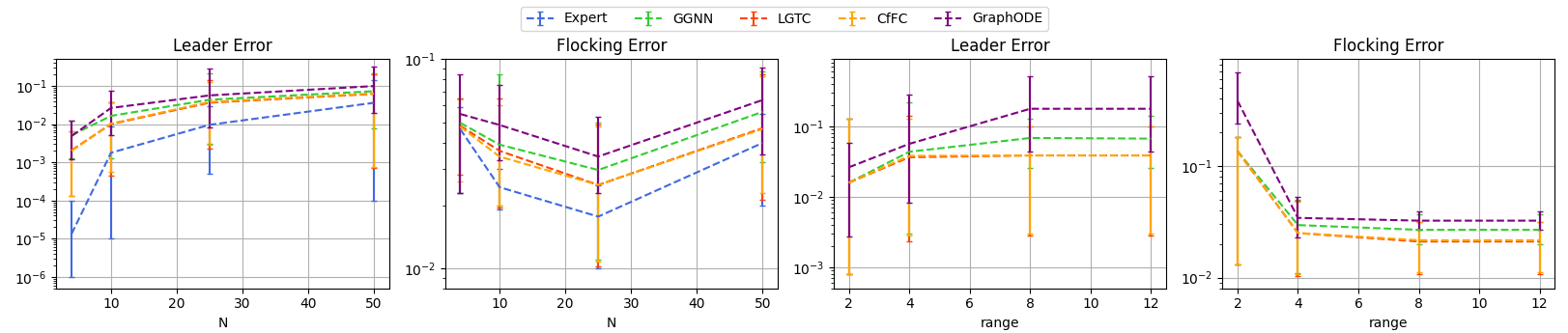}
\caption{Flocking and Leader Error for GGNN, LGCT, CfGC and GraphODE, varying the team size N with a communication range of $4$m and a variable communication range with $N=25$.}
\label{fig:flocking_results}
\end{figure*}
We collected a dataset by recording 60 trajectories,  further separated into three subsets of training, validation and test set using the proportion $70\%-10\%-20\%$, respectively. Each trajectory is generated by randomly positioning the agents in a square such that their inter-distance is between $0.6$~m and $1.0$~m and that their initial velocities are picked at random from the interval $[ -2, 2 ]$~m/s in each direction. The leader is randomly selected among the agents and the target position is randomly located within a square of length $20$~m centered at the location of the leader. Regardless of the target location, the trajectories have a duration of $2.5$~s and input saturation at $5$~m/s$^2$. Moreover, the $60$ trajectories are recorded with a random number of agents among $N=[4,6,10,12,15]$. We fixed the communication range to $R=4$~m and the sensing to $R_{CA}=1$~m. We trained the models for $120$ epochs and executed the DAGGER algorithm~\cite{ross2011reduction} every $20$ epochs. The algorithm evaluates the expert controller in~\eqref{eq:expert_flocking} on the enrolled state trajectories applying the learned control and adding them to the training set. Note that, thanks to the use of DAGGER, we do not need a large dataset. We solve the imitation learning problem using the ADAM algorithm~\cite{kingma2014adam} with a learning rate $1e-3$ and forgetting factors $0.9$ and $0.999$. The loss function used for imitation learning is the mean squared error between the output of the model and the optimal control action.
\subsection*{\textbf{Results}}
In Fig.~\ref{fig:flocking_results}, we show a comparison between GraphODE, GGNN, LGTC and CfGC controllers for the flocking control case. We evaluate the 4 controllers on 2 sets of experiments with $20$ trajectories each. In the experiments, we varied team size and communication range to test the robustness of the controllers. Figure~\ref{fig:flocking_results} reports the leader position error evaluated after a fixed time of $2.5s$ with respect to the leader starting location, i.e. $e_f/e_s$ with $e_f,e_s$ respectively being the final and the initial square distance of the leader from the target. We also show the average flocking error in the interval $[0,2.5s]$ in logarithmic scale. We assume that the communication happens within the sampling time $T$.

In the first experiment, the controller scalability is evaluated for team sizes $N=[4,10,25,50]$ with a communication range at $4$~m. Figure~\ref{fig:flocking_results} shows empirically the higher prediction abilities of LGTC and CfGC compared to GGNN, as they are up to $40\%$ for the flocking error and $10\%$ for the leader error closer to the expert controller. This is despite LGTC and CfGC communicating fewer state and input variables as explained in the previous section. Moreover, the performances of LGTC and CfGC are close to each other, confirming the good approximation of CfGC. Sometimes, for example, for the flocking error with $N=10$, we can see CfGC performing better than LGTC of $1-2\%$. This can be justified by the fact that CfGC is simpler to train, as it is a discrete model, and so it reduces the vanishing gradient phenomena. In general, GraphODE is the one that performs the worst both for leader and flocking errors. This can be explained easily by no direct use of the input features, as they are only used to initialize the internal state, resulting in an auto-regressive model.

In the second set of experiments, we evaluate the performances of the networks when the communication range differs from the training range $R=4$~m. For $2$~m all the models show a drop in the performances with $0.18$~m/s in the flocking error for GGNN,LGTC and CfGC, and $0.3$~m/s for GraphODE. This drop appears due to the lower ability to broadcast the information among the agents when they space out over $2$~m. As expected, when the communication range increases, the flocking errors decrease for all controllers, since they can communicate with more agents at the same time. However, in this case, the leader error increases, since the leader is often ``encapsulated" by the other team agents and it is thus forced to follow them to not collide, causing a slower convergence to the target. Note that this behaviour also affects the expert controller. Notably, also for this experiment, LGTC and CfCG have better performances compared to the others and are similar to each other.

\section{CONCLUSIONS}
\label{conclusions}
In this work, we model a new continuous-time graph neural network, Liquid-Graph Time-constant (LGTC). Its enhanced prediction capabilities allow us to reduce communication load, one of the main drawbacks of GNNs. We analyze its dynamic properties leveraging contraction analysis and find a closed-form approximation (CfGC) that preserves the dynamic properties of the system. We validate the proposed model in a flocking control case and compare it with $\delta$ISS GGNN and the GraphODE. Results suggest that, despite the reduced communicated vector, LGTC and CfGC have performances closer to the centralized expert and more robust to parametric changes in a deployment scenario, such as communication radius and team size. Future works will consider how to further reduce the communication load, still high compared to model-based approaches, as well as an experimental validation on quadrotor UAVs.

\bibliographystyle{IEEEtran} 
\bibliography{IEEEabrv,bibliografy}

\end{document}